\documentclass{article}

\usepackage{microtype}
\usepackage{graphicx}
\usepackage{subfigure}
\usepackage{booktabs} 
\usepackage{amsmath}
\usepackage{amsfonts}
\usepackage{amsthm}
\usepackage[textsize=scriptsize]{todonotes}
\usepackage{bbm}

\usepackage{hyperref}

\newcommand{\algname}{\textsc{AutoVC }}
\newcommand{\algnamens}{\textsc{AutoVC}}

\newcommand{\algnameonehotns}{\textsc{AutoVC-one-hot}}

\newtheorem{theorem}{Theorem}

\newtheorem{lemma}{Lemma}[theorem]

\DeclareMathOperator*{\plim}{plim}

\usepackage[accepted]{icml2019}

\icmltitlerunning{\algnamens:  Zero-Shot Voice Style Transfer with Only Autoencoder Loss}

\begin{document}

\twocolumn[
\icmltitle{\algnamens:  Zero-Shot Voice Style Transfer with Only Autoencoder Loss}



\icmlsetsymbol{equal}{*}

\begin{icmlauthorlist}
\icmlauthor{Kaizhi Qian}{equal,uiuc}
\icmlauthor{Yang Zhang}{equal,mitibm,ibm}
\icmlauthor{Shiyu Chang}{mitibm,ibm}
\icmlauthor{Xuesong Yang}{uiuc}
\icmlauthor{Mark Hasegawa-Johnson}{uiuc}
\end{icmlauthorlist}

\icmlaffiliation{mitibm}{MIT-IBM Watson AI Lab, Cambridge, MA, USA}
\icmlaffiliation{ibm}{IBM Research, Cambridge, MA, USA}
\icmlaffiliation{uiuc}{University of Illinois at Urbana-Champaign, IL, USA}

\icmlcorrespondingauthor{Kaizhi Qian}{kqian3@illinois.edu}

\icmlkeywords{Style Transfer, Voice Conversion, Autoencoder, Zero-Shot Learning}

\vskip 0.3in
]



\printAffiliationsAndNotice{\icmlEqualContribution} 

\begin{abstract}
Non-parallel many-to-many voice conversion, as well as zero-shot voice conversion, remain under-explored areas. Deep style transfer algorithms, such as generative adversarial networks (GAN) and conditional variational autoencoder (CVAE), are being applied as new solutions in this field. However, GAN training is sophisticated and difficult, and there is no strong evidence that its generated speech is of good perceptual quality. On the other hand, CVAE training is simple but does not come with the distribution-matching property of a GAN. In this paper, we propose a new style transfer scheme that involves only an autoencoder with a carefully designed bottleneck. We formally show that this scheme can achieve distribution-matching style transfer by training only on a self-reconstruction loss. Based on this scheme, we proposed \algnamens, which achieves state-of-the-art results in many-to-many voice conversion with non-parallel data, and which is the first to perform zero-shot voice conversion.

\end{abstract}

\section{Introduction}
\label{sec:intro}
The idea of speaking in someone else's voice never fails to be a fascinating element in action and fiction movies, and it also finds its way to many practical applications, \emph{e.g.} privacy and identity protection, creative industry \emph{etc.}  In the speech research community, this task is referred to as the voice conversion problem, which involves modifying a given speech from a source speaker to match the vocal qualities of a target speaker.

Despite the continuing research efforts in voice conversion, three problems remain under-explored. First, most voice conversion systems assume the availability of parallel training data, \emph{i.e.} speech pairs where the two speakers utter the same sentences.  Only a few can be trained on non-parallel data. Second, among the few existing algorithms that work on non-parallel data, even fewer can work for many-to-many conversion, \emph{i.e.} converting from multiple source speakers to multiple target speakers. Last but not least, no voice conversion systems are able to perform zero-shot conversion, \emph{i.e.} conversion to the voice of an unseen speaker by looking at only a few of his/her utterances. 

With the recent advances in deep style transfer, the traditional voice conversion problem is being recast as a style transfer problem, where the vocal qualities can be regarded as styles, and speakers as domains. There are many style transfer algorithms that do not require parallel data, and are applicable to multiple domains, so they are readily available as new solutions to voice conversion. In particular, generative adversarial network (GAN) \cite{goodfellow2014generative} and conditional variational autoencoder (CVAE) \cite{kingma2013auto, kingma2014semi}, are gaining popularity in voice conversion.

However, neither GAN nor CVAE is perfect. GAN comes with a nice theoretical justification that the generated data would match the distribution of the true data, and has achieved state-of-the-art results, particularly in computer vision.  However, it is widely acknowledged that GAN is very hard to train, and its convergence property is fragile.  Also, although there is an increasing number of works that introduce GAN to speech generation \cite{donahue2018adversarial} and speech domain transfer \cite{pascual2017segan, subakan2018generative, fan2018svsgan, hosseini2018multi}, there is no strong evidence that the generated speech \textbf{\textit{sounds}} real.  Speech that is able to fool the discriminators has yet to fool human ears.  On the other hand, CVAE is easier to train.  All it needs to do is to perform self-reconstruction and maximize a variational lower bound of the output probability. The intuition is to infer a hypothetical style-independent hidden variable, which is then combined with the new style information to generate the style-transferred output. However, CVAE alone does not guarantee distribution matching, and often suffers from over-smoothing of the conversion output \cite{kameoka2018stargan}.

Due to the lack of a suitable style transfer algorithm, existing voice conversion systems have yet to produce satisfactory results, which naturally leads to the following formulation of the problem. Is there a style transfer algorithm that can be proven to match the distribution as GAN does, that trains as easily as CVAE, and that works better for speech?

Motivated by this, in this paper, we propose a new style transfer scheme, which involves only a \textbf{\textit{vanilla}} autoencoder with a carefully designed bottleneck.  Similar to CVAE, the proposed scheme only needs to be trained on the self-reconstruction loss, but it has a distribution matching property similar to GAN's. This is because the correctly-designed bottleneck will learn to remove the style information from the source and get the style-independent code, which is the goal of CVAE, but which the training scheme of CVAE is unable to guarantee.

Based on this scheme, we propose \algnamens, a many-to-many voice style transfer algorithm without parallel data. \algname follows the autoencoder framework and is trained only on autoencoder loss, but it introduces carefully-tuned dimension reduction and temporal downsampling to constrain the information flow. As we will show, this simple scheme leads to a significant performance gain. \algname achieves superior performance on a traditional many-to-many conversion task, where all the speakers are seen in the training set. Also, equipped with a speaker embedding trained for speaker verification \cite{heigold2016end, wan2018generalized}, \algname is among the first to perform zero-shot voice conversion with decent performance. Considering the quality of the results and the simplicity of its training scheme, \algname opens a new path towards a simpler and better voice conversion and general style transfer systems. The implementation will become publicly available. 

\section{Related Works}
\label{sec:realted}
There are several works that perform non-parallel many-to-many voice conversion using VAE and its combination with adversarial training. VAE-VC \cite{hsu2016voice} is a simple voice conversion system using VAE. Afterward, much research focuses on removing the style information from the VAE code. VAW-GAN \cite{hsu2017voice} introduces a GAN on the VAE output. CDVAE-VC \cite{huang2018voice} introduces two VAEs on two spectral features and forces the latent codes of the two features to contain similar information. ACVAE-VC \cite{kameoka2018acvae} introduces an auxiliary classifier on the output to encourage the conversion results to be correctly classified as the target speaker's utterances. \citet{chou2018multi} introduce a classifier on the code and a GAN on the output. Similarly, StarGAN \cite{kaneko2017parallel} and CycleGAN \cite{zhu2017unpaired}, which consist of encoder-decoder architectures with GAN, are applied to voice conversion \cite{kameoka2018stargan, fang2018high}. GAN alone is also applied to voice conversion \cite{gao2018voice}. However, the conversion quality of these algorithms is still limited. Text transcriptions are introduced to assist the learning of the latent code \cite{xie2016kl, saito2018non, biadsy2019parrotron}, but we will focus on voice conversion without text transcriptions, which is more flexible for low-resourced languages.

\citet{atalla2019look, chou2018multi, nachmani2019unsupervised} conduct research on style transfer using autoencoder, but none has unveiled its distribution-matching property by properly designing the bottleneck.

\section{Style Transfer Autoencoder}
\label{sec:framework}
In this section, we will discuss how and why an autoencoder can match the data distribution as GAN does. Although our intended application is voice conversion, the discussion in this section is applicable to other style transfer applications as well.  As general mathematical notations, upper-case letters, \emph{e.g.} $X$, denote random variables/vectors; lower-case letters, \emph{e.g.} $x$, denote deterministic values or instances of random variables; $X(1:T)$ denotes a random process, with $(1:T)$ denoting a collection of time indices running from $1$ to $T$. For notational ease, sometimes the time indices are omitted to represent the collection of the random process at all times. $p_X(\cdot | Y)$ denotes the probability mass function (PMF) or probability density function (PDF) of $X$ conditional on $Y$; $p_X(\cdot | Y=y)$, or sometimes $p_X(\cdot | y)$ without causing confusions, denotes the PMF/PDF of $X$ conditional on $Y$ taking a specific value $y$; similarly, $\mathbb{E}[X|Y]$, $\mathbb{E}[X|Y=y]$ and $\mathbb{E}[X|y]$ denote the corresponding conditional expectations. It is worth mentioning that $\mathbb{E}[X|Y]$ is still a random, but $\mathbb{E}[X|Y=y]$ or $\mathbb{E}[X|y]$. $H(\cdot)$ denotes the entropy, and $H(\cdot | \cdot)$ denotes the conditional entropy. 

\subsection{Problem Formulation}
\label{subsec:formulation}

Assume that speech is generated by the following stochastic process. First, a speaker identity $U$ is a random variable drawn from the speaker population $p_U(\cdot)$. Then, a content vector $Z=Z(1:T)$ is a random process drawn from the joint content distribution $p_Z(\cdot)$. Here content refers to the phonetic and prosodic information. Finally, given the speaker identity and content, the speech segment $X=X(1:T)$ is a random process randomly sampled from the speech distribution, \emph{i.e.} $p_{X}(\cdot| U, Z)$, which characterizes the distribution of the speaker $U$'s speech uttering the content $Z$. $X(t)$ can represents a sample of a speech waveform, or a frame of a speech spectrogram. In this paper, we will be working on speech spectrograms. Here, we assume that each speaker produces the same amount of gross information, \emph{i.e.}
\begin{equation}
\small
    H(X | U = u) = h_\textrm{speech} = \textrm{constant},
    \label{eq:constant_info}
\end{equation}
regardless of $u$.

Now, assume two sets of variables, $(U_1, Z_1, X_1)$ and $(U_2, Z_2, X_2)$, are independent and identically distributed (i.i.d.) random samples generated from this process. $(U_1, Z_1, X_1)$ belong to the \emph{source speaker} and $(U_2, Z_2, X_2)$ belong to the \emph{target speaker}. Our goal is to design a speech converter that produces the conversion output, $\hat{X}_{1\rightarrow 2}$, which preserves the content in $X_1$, but matches the speaker characteritics of speaker $U_2$. Formally, an ideal speech converter should have the following desirable property:
\begin{equation}
\small
\begin{aligned}
 p_{\hat{X}_{1\rightarrow 2}}(\cdot | U_2 = u_2, Z_1 = z_1) = p_X(\cdot | U=u_2, Z=z_1).
    \label{eq:ideal}
\end{aligned}
\end{equation}
Eq.~\eqref{eq:ideal} means that given the target speaker's identity $U_2=u_2$ and the content in the source speech $Z_1 = z_1$, the converted speech should sound like $u_2$ uttering $z_1$.

When $U_1$ and $U_2$ are both seen in the training set, the problem is a standard multi-speaker conversion problem, which has been addressed by some existing works. When $U_1$ or $U_2$ is not included in the training set, the problem becomes the more challenging zero-shot voice conversion problem, which is also a target task of the proposed \algnamens.  

\begin{figure}[t!]
    \centering
    \subfigure[Conversion]{\includegraphics[width=0.48\columnwidth]{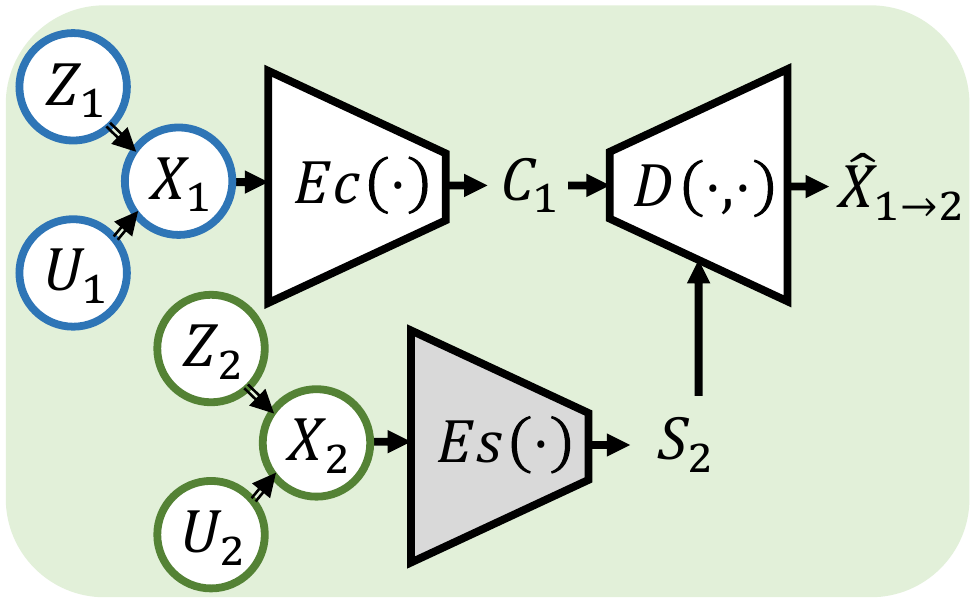}}
    \subfigure[Training]{\includegraphics[width=0.48\columnwidth]{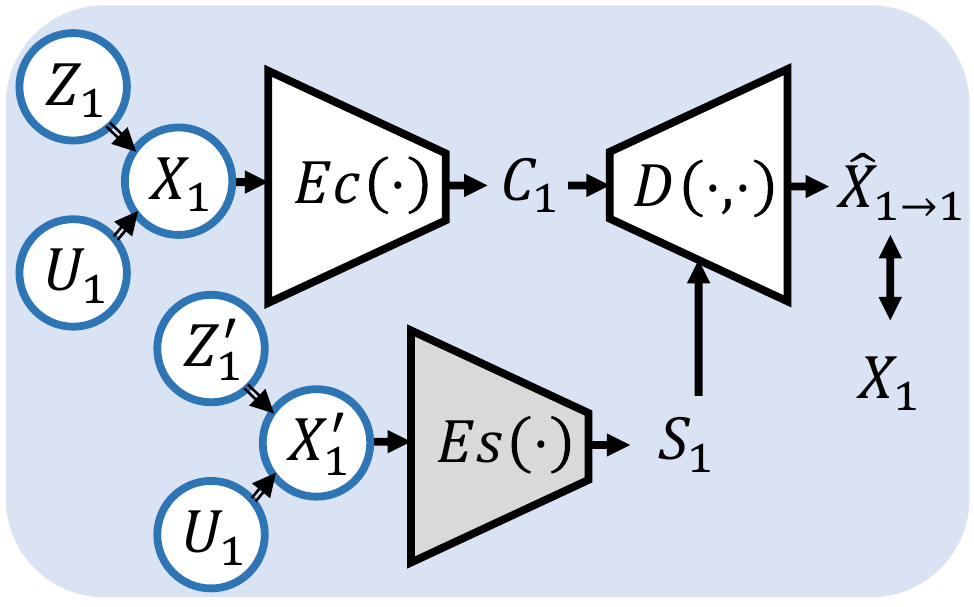}}
    \vspace*{-0.1in}
    \caption{The style transfer autoencoder framework. The ovals denote the probabilistic graphical model of the speech generation process. The grey boxes denote pre-trained modules. (a) During conversion, the source speech is fed to the content encoder. An utterance of the target speaker is fed to the speaker encoder. The decoder produces the conversion results. (b) During training, the source speech is fed to the content encoder. Another utterance of the same \emph{source} speaker is fed to the speaker encoder. The content encoder and the decoder minimize the self-reconstruction error.}
    \label{fig:framework}
    \vspace*{-0.05in}
\end{figure}

\subsection{The Autoencoder Framework}

\algname solves the voice conversion problem with a very simple autoencoder framework, as shown in Fig.~\ref{fig:framework}. The framework consists of three modules, a content encoder $E_c(\cdot)$ that produces a content embedding from speech, a speaker encoder $E_s(\cdot)$ that produces a speaker embedding from speech, and a decoder $D(\cdot, \cdot)$ that produce speech from content and speaker embeddings. The inputs to these modules are different for conversion and training. 

\paragraph{Conversion:} As shown in Fig.~\ref{fig:framework}(a), during the actual conversion, the source speech $X_1$ is fed into the content encoder to have content information extracted. The target speech is fed into the speaker encoder to provide target speaker information. The decoder produces the converted speech based on the content information in the source speech and the speaker information in the target speech.

\begin{equation}
\small
\begin{aligned}
C_1 = E_c(X_1),  \quad S_2 = E_s(X_2), \quad \hat{X}_{1\rightarrow2} = D(C_1, S_2).
\end{aligned}
\label{eq:convert}
\end{equation}
Here $C_1$ and $\hat{X}_{1\rightarrow2}$ are both random processes. $S_2$ is simply a random vector.

\paragraph{Training:} Throughout the paper, we will assume the speaker encoder is already pre-trained to extract some form of speaker dependent embedding, so by training we refer to the training of the content encoder and the decoder. As shown in Fig.~\ref{fig:framework}(b), since we do not assume the availability of parallel data, only self-reconstruction is needed for training. More specifically, the input to the content encoder is still $X_1$, but the input to the style encoder becomes an utterance from the same speaker $U_1$, denoted as $X'_1$.\footnote{$X_1'$ and $X_1$ can be the same or different.} Then for each input speech $X_1$, \algname learns to reconstruct itself: 
\begin{equation}
\small
\begin{aligned}
C_1 = E_c(X_1),  \quad S_1 = E_s(X'_1), \quad \hat{X}_{1\rightarrow1} = D(C_1, S_1).
\end{aligned}
\label{eq:training}
\end{equation}

The loss function to minimize is simply the weighted combination of the self-reconstruction error and the content code reconstruction error, \emph{i.e.}

\begin{equation}
\small
    \min_{E_c(\cdot), D(\cdot, \cdot)} L = L_\textrm{recon} + \lambda L_\textrm{content},
    \label{eq:loss}
\end{equation}
where
\begin{equation}
\small
    \begin{aligned}
    &L_\textrm{recon} = \mathbb{E} [\lVert \hat{X}_{1\rightarrow1} - X_1 \rVert_2^2], \\
    &L_\textrm{content} = \mathbb{E} [\lVert E_c(\hat{X}_{1\rightarrow1}) - C_1 \rVert_1].
    \end{aligned}
    \label{eq:loss_detail}
\end{equation}
As it turns out, this simple training scheme is sufficient to produce the ideal distribution-matching voice conversion, as will be shown in the next section. 

\subsection{Why does it work?}
\begin{figure*}
\centering
\subfigure[Bottleneck too wide]{\includegraphics[width=0.23\textwidth]{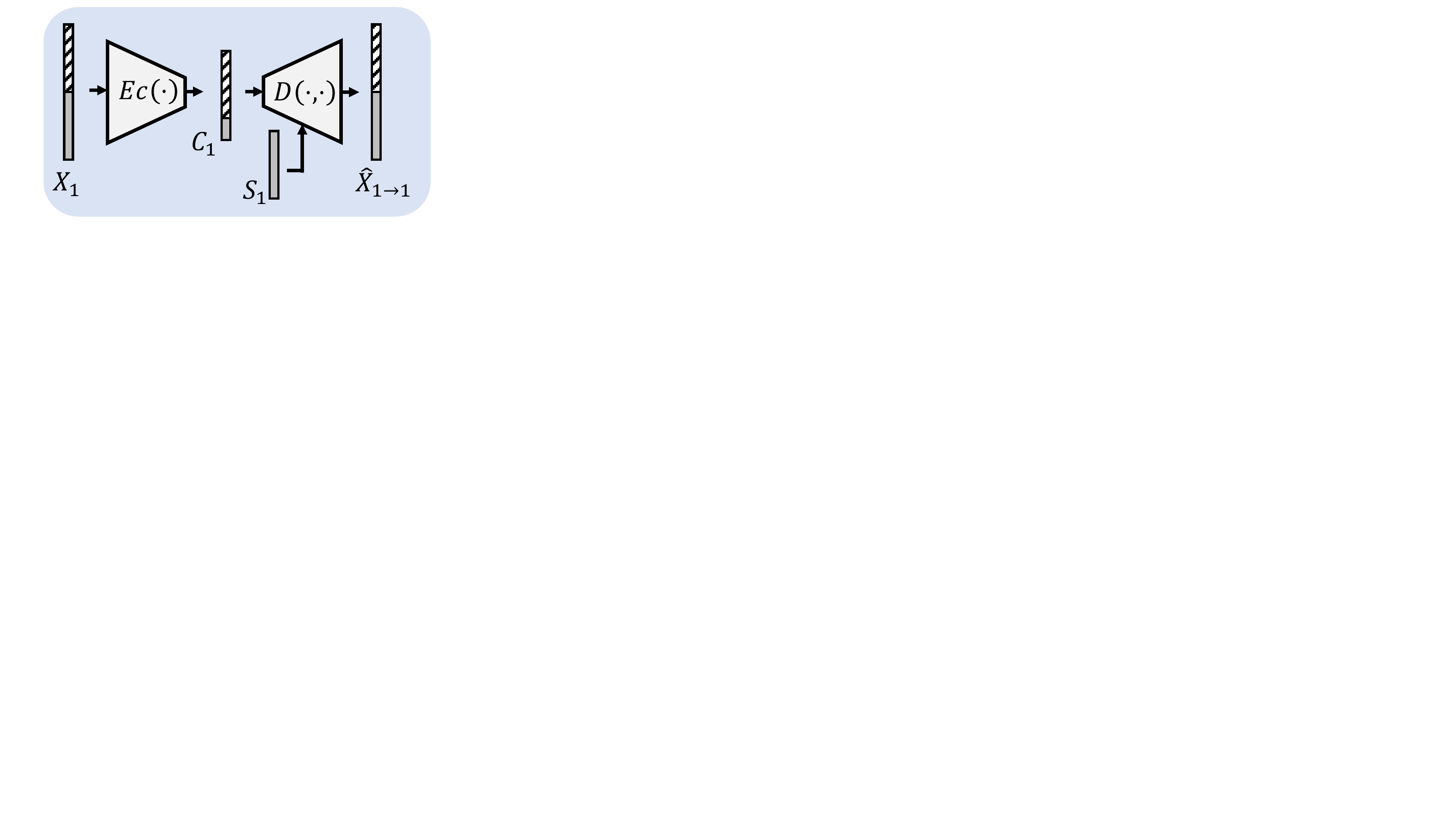}}
\subfigure[Bottleneck too narrow]{\includegraphics[width=0.23\textwidth]{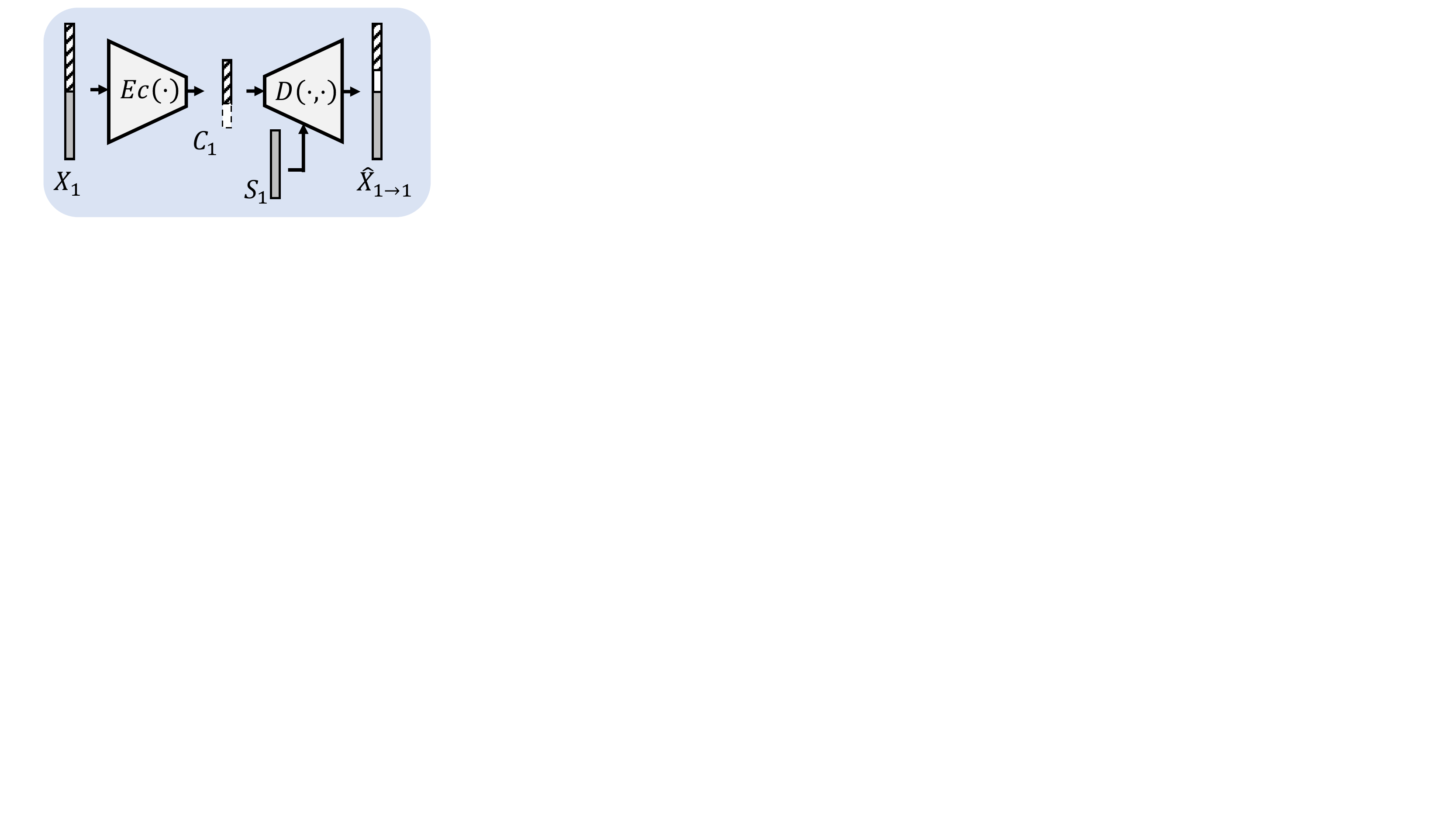}}
\subfigure[Bottleneck just right]{\includegraphics[width=0.23\textwidth]{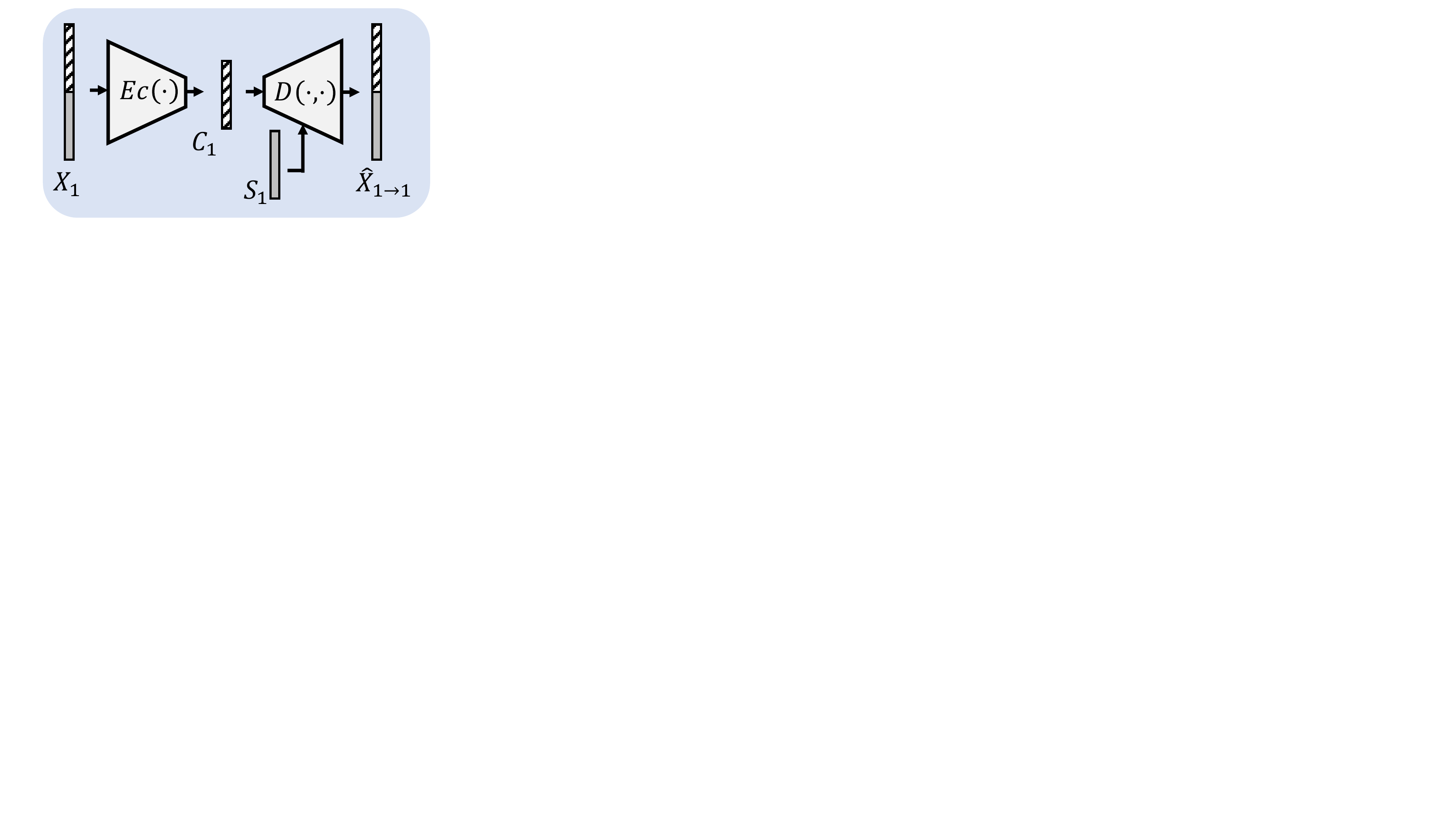}}
~
\subfigure[Conversion]{\includegraphics[width=0.26\textwidth]{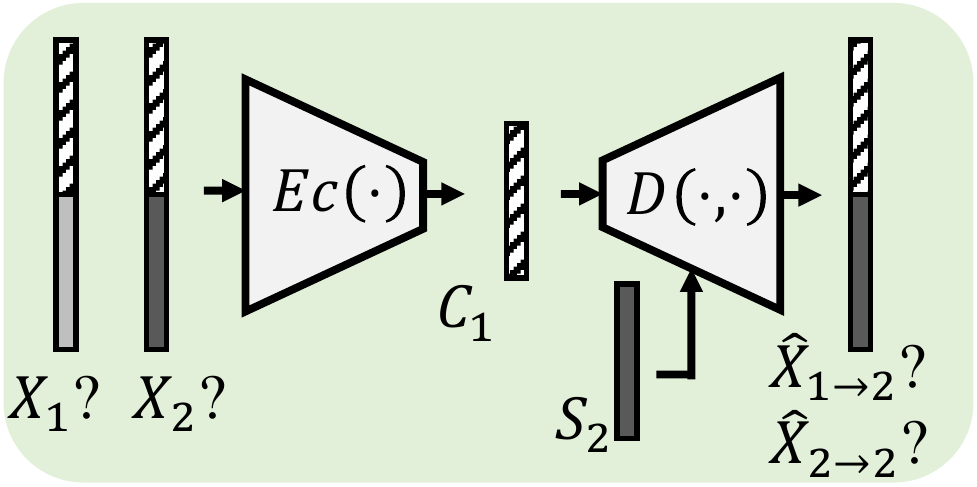}}
\vspace*{-0.1in}
\caption{An intuitive explanation of how \algname works. The target speaker is the same as the source speaker during training ((a)-(c)), and different during the actual conversion ((d)). Each speech segment contains two types of information: the speaker information (solid) and content information (striped). (a) When the bottleneck is too wide, the content embedding will contain some source speaker information. (b) When the bottleneck is too narrow, the content information is lost, which leads to imperfect reconstruction. (c) When the bottleneck is just right, perfect reconstruction is achievable, \emph{and} the content embedding contains no source speaker information. (d) During the actual conversion, the output should contain no information about the source speaker, so the conversion quality should be as high as if it were doing self-reconstruction.}
\label{fig:explain}
\vspace*{-0.1in}
\end{figure*}

We will formally show this autoencoder-based training scheme is able to achieve ideal voice conversion (Eq.~\eqref{eq:ideal}). The secret recipe is to have a proper information bottleneck. We will first state the theoretical guarantee and then present an intuitive explanation.

The following theorem characterizes the theoretical guarantee of our proposed framework.






\begin{theorem}
Consider the autoencoder framework depicted in Eqs.~\eqref{eq:convert} and \eqref{eq:training}. Given the following assumption:

1. The speaker embedding of different utterances of the same speaker is the same. Formally, if $U_1 = U_2$,  $E_s(X_1) = E_s(X_2)$.

2. The speaker embedding of different speakers is different. Formally, if $U_1 \neq U_2$, $E_s(X_1) \neq E_s(X_2)$.

3. $\{X_1(1:T)\}$ is an ergodic stationary order-$\tau$ Markov process with bounded second moment, \emph{i.e.}
\begin{equation}
\small
\begin{aligned}
p_{X_1(t)}( \cdot | X_1(1:t-1), U_1) = p_{X_1(t)}(\cdot | X_1(t-\tau:t-1), U_1).    
\end{aligned}
\end{equation}
Further assume $X_1$ has finite cardinality.

4. Denote $n$ as the dimension of $C_1$. Then $n = \lfloor n^* + T^{2/3} \rfloor$,
where $n^*$ is the optimal coding length of $p_{X_1}(\cdot | U_1)$\footnote{From the assumpion in Eq.~\eqref{eq:constant_info}, $n^*$ is assumed to be a constant regardless of $U_1$}.

Then the following holds. For each $T$, there exists a content encoder $E^*_c(\cdot; T)$ and a decoder $D^*(\cdot, \cdot; T)$, \emph{s.t.} $\lim_{T \rightarrow \infty} L = 0$, 
and
\begin{equation}
\small
\begin{aligned}
    \lim_{T \rightarrow \infty} \frac{1}{T} KL~(p_{\hat{X}_{1\rightarrow 2}}(\cdot | u_2, z_1) || p_X(\cdot | U=u_2, Z=z_1)) 
    = 0,
\end{aligned}
\end{equation}
where $KL(\cdot || \cdot)$ denotes the KL-divergence.

\label{thm:guarantee}
\end{theorem}

The conclusion of Thm.~\ref{thm:guarantee} can be interpreted as follows. If the number of frames $T$ is large enough, and if the bottleneck dimension $n$ is properly set, then the global optimizer of the loss function in Eq.~\eqref{eq:loss} would approximately satisfy the ideal conversion property in Eq.~\eqref{eq:ideal}. This conclusion is quite strong, because a major justification of applying GAN to style transfer, despite all its hassles, is that it can ideally match the distribution of the true samples from the target domain. Now Thm.~\ref{thm:guarantee} conveys the following message: to achieve the desired distribution matching, an autoencoder is all you need.  The formal proof of Thm.~\ref{thm:guarantee} will be presented in the appendix. Here, we will present an intuitive explanation, which is also the gist of our proof. The basic idea is that the bottleneck dimension of the content encoder needs to be set such that it is just enough to code the speaker independent information.

As shown in Fig.~\ref{fig:explain}, speech contains two types of information: the speaker information (shown as solid color) and the speaker-independent information (shown as striped), which we will refer to as the content information\footnote{The speaker-independent information includes but is not limited to the content information in $Z$, but for convenience, we will refer to the speaker-independent information as content information.}. Suppose the bottleneck is very wide, as wide as the input speech $X_1$. The most convenient way to do self-reconstruction is to copy $X_1$ as is to the content embedding $C_1$, and this will guarantee a perfect reconstruction. However as the dimension of $C_1$ decreases, $C_1$ is forced to lose some information. Since the autoencoder attempts to achieve perfect reconstruction, it will choose to lose speaker information because the speaker information is already supplied in $S_1$. In this case, perfect reconstruction is still possible, but the $C_1$ may contain some speaker information, as shown in Fig.~\ref{fig:explain}(a).

On the other hand, if the bottleneck is very narrow, then the content encoder will be forced to lose so much information that not only the speaker information but also the content information is lost. In this case, the perfect reconstruction is impossible, as shown in Fig.~\ref{fig:explain}(b).

Therefore, as shown in Fig.~\ref{fig:explain}(c), when the dimension of $C_1$ is chosen such that the dimension reduction is just enough to get rid of all the speaker information but no content information is harmed, we have reached our desirable condition, under which two important properties hold:

1. Perfect reconstruction is achieved.

2. The content embedding $C_1$ does not contain any information about the source speaker $U_1$, which we refer to as \emph{speaker disentanglement}.

We will now show by contradiction how these two properties imply an ideal conversion. Suppose when \algname is performing an actual conversion (source and target speakers are different), the quality is low, or does not sound like the target speaker at all. By property 1, we know that the reconstruction (source and target speakers are the same) quality is high. However, according to Eq.~\eqref{eq:convert}, the output speech $\hat{X}_{1\rightarrow2}$ can only access $C_1$ and $S_2$, both of which do not contain any information of the source speaker $U_1$. In other words, from the conversion output, one can never tell if it is produced by self-reconstruction or conversion, as shown in Fig.~\ref{fig:explain}(d). If the conversion quality is low, but the reconstruction quality is high, one will be able to distinguish between conversion and reconstruction above chance, which leads to a contradiction.

\section{\algname Architecture}
\label{sec:archi}
\begin{figure*}[t!]
\centering
\includegraphics[width=0.9\textwidth]{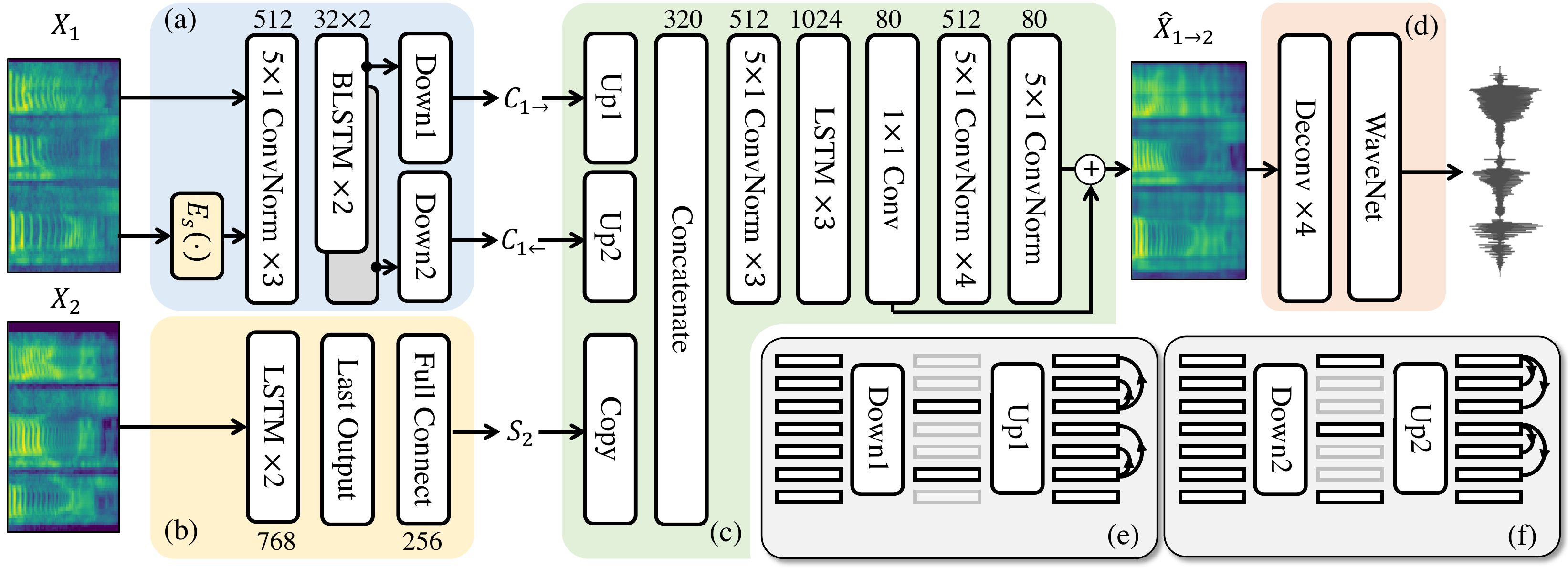}
\vspace*{-0.05in}
\caption{\algname architecture. The number above each block represents the cell/output dimension of the structure. ConvNorm denotes convolution followed by batch normalization. BLSTM denotes bi-directional LSTM, whose white block denotes forward direction, and grey block denotes backward direction. (a) The content encoder. The $E_s(\cdot)$ module is of the same architecture as in (b). (b) The style encoder. (c) The decoder. (d) The spectrogram inverter. (e) and (f) demonstrate the downsampling and upsampling of the forward and backward outputs of the Bi-directional LSTM, using a up/downsampling factor of $3$ as an example. The real up/downsampling factor is $32$. The lightened feature denotes that they are removed; the arrows denote copying the feature at the arrow origin to the destination.}
\label{fig:archi}
\end{figure*}

As shown in Fig.~\ref{fig:archi}, \algname consists of three major modules: a speaker encoder, a content encoder, a decoder. \algname works on the speech mel-spectrogram of size $N$-by-$T$, where $N$ is the number of mel-frequency bins and $T$ is the number of time steps (frames). A spectrogram inverter is introduced to convert the output mel-spectrogram back to the waveform, which will also be detailed in this section.

\subsection{The Speaker Encoder}
According to assumptions 1 and 2 in Thm.~\ref{thm:guarantee}, the goal of the speaker encoder is to produce the same embedding for different utterances of the same speaker, and different embeddings for different speakers. For conventional many-to-many voice conversion, the one-hot encoding of speaker identities suffices. However, in order to perform zero-shot conversion, we need to apply an embedding that is generalizable to unseen speakers. Therefore, inspired by \cite{jia2018transfer}, we follow the design in \cite{wan2018generalized}. As shown in Fig.~\eqref{fig:archi}(b), the speaker encoder consists of a stack of two LSTM layers with cell size 768. Only the output of the last time is selected and projected down to dimension 256 with a fully connected layer. The resulting speaker embedding is a $256$-by-$1$ vector. The speaker encoder is pre-trained on the GE2E loss \cite{wan2018generalized} (the softmax loss version), which maximizes the embedding similarity among different utterances of the same speaker, and minimizes the similarity among different speakers. Therefore, it is very consistent with assumptions 1 and 2 in Thm.~\ref{thm:guarantee}. 

In our implementation, the speaker encoder is pre-trained on the combination of VoxCeleb1 \cite{voxceleb} and Librispeech \cite{librispeech} corpora, where there are a total of 3549 speakers.

\subsection{The Content Encoder}

As shown in Fig.~\ref{fig:archi}(a), the input to the content encoder is the 80-dimensional mel-spectrogram of $X_1$ concatenated with the speaker embedding, $E_s(X_1)$, at each time step. The concatenated features are fed into three $5\times 1$ convolutional layers, each followed by batch normalization and ReLU activation. The number of channels is $512$. The output then passes to a stack of two bidirectional LSTM layers. Both the forward and backward cell dimensions are $32$, so their combined dimension is $64$.

As a key step of constructing the information bottleneck, both the forward and backward outputs of the bidirectional LSTM are downsampled by $32$. The downsampling is performed differently for the forward and backward paths. For the forward output, the time steps $\{0, 32, 64, \cdots \}$ are kept; for the backward output, the time steps $\{31, 63, 95, \cdots \}$ are kept. Figs.~\ref{fig:archi}(e) and (f) also demonstrate how the downsampling is performed (for the ease of demonstration, the downsampling factor is set to $3$). The resulting content embedding is a set of two $32$-by-$T/32$ matrices, which we will denote $C_{1\rightarrow}$ and $C_{1\leftarrow}$ respectively. The downsampling can be regarded as dimension reduction along the temporal axis, which, together with the dimension reduction along the channel axis, constructs the information bottleneck.

\subsection{The Decoder}
The architecture of the decoder is inspired by \cite{shen2018natural}; and is shown in Fig.~\ref{fig:archi}(c). First, the content and speaker embeddings are both upsampled by copying to restore to the original temporal resolution. Formally, denotes the upsampled features as $U_\rightarrow$ and $U_\leftarrow$ respectively. Then
\begin{equation}
\small
\begin{aligned}
    U_\rightarrow(:, t) &= C_{1\rightarrow}(:, \lfloor t/32 \rfloor) \\
    U_\leftarrow(:, t) &= C_{1\leftarrow}(:, \lfloor t/32 \rfloor),
\end{aligned}
\end{equation}
where $(:, t)$ denotes indexing the $t$-th column. Figs.~\ref{fig:archi}(e) and (f) also demonstrate the copying. The underlying intuition is that each embedding at each time step should contain both past and future information. For the speaker embedding, simply copy the vector $T$ times. 

Then, the upsampled embeddings are concatenated and fed into three $5\times 1$ convolutional layers with $512$ channels, each followed by batch normalization and ReLU, and then three LSTM layers with cell dimension $1024$. The outputs of the LSTM layer are projected to dimension $80$ with a $1\times 1$ convolutional layer. This projection output is the initial estimate of the converted speech, denoted as $\tilde{X}_{1\rightarrow2}$.

In order to construct the fine details of the spectrogram better on top of the initial estimate, we introduce a post-network after the initial estimate, as introduced in \citet{shen2018natural}. The post network consists of five $5\times 1$ convolutional layers, where batch normalization and hyperbolic tangent are applied to the first four layers. The channel dimension for the first four layers is $512$, and goes down to $80$ in the final layer. We will refer to the output of the post-network as the residual signal, denoted as $R_{1\rightarrow2}$. The final conversion result is produced by adding the residual to the initial estimate, \emph{i.e.}
\begin{equation}
\small
    \hat{X}_{1\rightarrow2} = \tilde{X}_{1\rightarrow 2} + R_{1\rightarrow 2}.
\end{equation}

During training, reconstruction loss is applied to both the initial and final reconstruction results. Formally, in addition to the loss specified in Eq.~\eqref{eq:loss}, we add an initial reconstruction loss defined as
\begin{equation}
\small
    L_\textrm{recon0} = \mathbb{E} [\lVert \tilde{X}_{1\rightarrow1} - X_1 \rVert_2^2],
    \label{eq:loss_init_recon}
\end{equation}
where $\tilde{X}_{1\rightarrow1}$ is the reciprocal of $\tilde{X}_{1\rightarrow2}$ in the reconstruction case, \emph{i.e.} when $U_2 = U_1$. The total loss becomes
\begin{equation}
\small
    \min_{E_c(\cdot), D(\cdot, \cdot)} L = L_\textrm{recon} + \mu L_\textrm{recon0} + \lambda L_\textrm{content}.
    \label{eq:loss_new}
\end{equation}
Although Eq.~\eqref{eq:loss_new} deviates from Eq.~\eqref{eq:loss}, on which Thm.~\ref{thm:guarantee} rests, we found empirically that this improves convergence and does not harm the performance.

\subsection{The Spectrogram Inverter}
We apply the WaveNet vocoder as introduced in \citet{van2016wavenet}, which consists of four deconvolution layers. In our implementation, the frame rate of the mel-spetrogram is 62.5 Hz and the sampling rate of speech waveform is 16 kHz. So the deconvolution layers will upsample the spectrogram to match the sampling rate of the speech waveform. Then, a standard 40-layer WaveNet conditioning upon the upsampled spectrogram is applied to generate the speech waveform. We pre-trained the WaveNet vocoder using the method described in \citet{shen2018natural} on the VCTK corpus.

\section{Experiments}
\label{sec:exper}


%

In this section, we will evaluate \algname on many-to-many voice conversion tasks, and empirically validate the assumptions of the \algname framework. We strongly encourage readers to listen to the demos\footnote{https://auspicious3000.github.io/autovc-demo/}.


\subsection{Configurations}

The evaluation is performed on the VCTK corpus \cite{veaux2016superseded}, which contains 44 hours of utterances from 109 speakers. Each speaker reads a different set of sentences, except for the rainbow passage\footnote {http://web.ku.edu/~idea/readings/rainbow.htm} and the elicitation paragraph. So the conversion setting is non-parallel. Depending on the conversion tasks, different subsets of speakers were selected. The data of each speaker is then partitioned into training and test sets by $9$:$1$. \algname is trained with a batch size of two for 100k steps, using the ADAM optimizer. The speaker embedding is generated by feeding $10$ two-second utterances of the same speaker to the speaker encoder and averaging the resulting embeddings. The weights in Eq.~\eqref{eq:loss_new} are set to $\lambda = 1$, $\mu = 1$. 

We performed two subjective tests on Amazon Mechanical Turk (MTurk)\footnote{https://www.mturk.com/}. In the first test, called the mean opinion score (MOS) test, the subjects are presented with converted utterances. For each utterance, the subjects are asked to assign a score of $1$-$5$ on the naturalness on the converted speech. In the second test, called the similarity test, the subjects are presented with pairs of utterances. In each pair, there is one converted utterance, and one utterance from the target speaker uttering the same sentence. For each pair, the subjects are asked to assign a score of $1$-$5$ on the voice similarity. We follow the design in \citet{wester2016analysis} to cue the subjects to judge if the speakers are the same, and how confident they are with their judgment. Thus the similarity score of $5$ corresponds to the same speaker with high confidence, and $1$ corresponds to different speakers with high confidence. The subjects are explicitly asked to focus on the voice rather than intonation and accent.

\subsection{Traditional Many-to-Many Conversion}
\label{subsec:tradition}

\begin{figure}[!t]
\centering
\includegraphics[width=0.8\columnwidth]{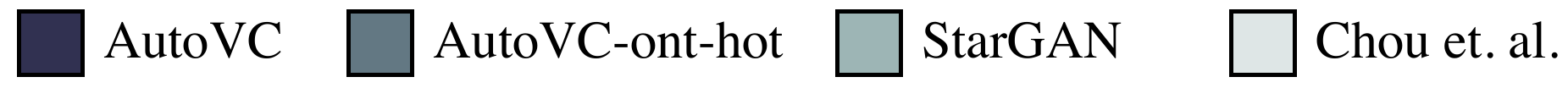}
\subfigure[MOS]{\includegraphics[width=0.5\linewidth]{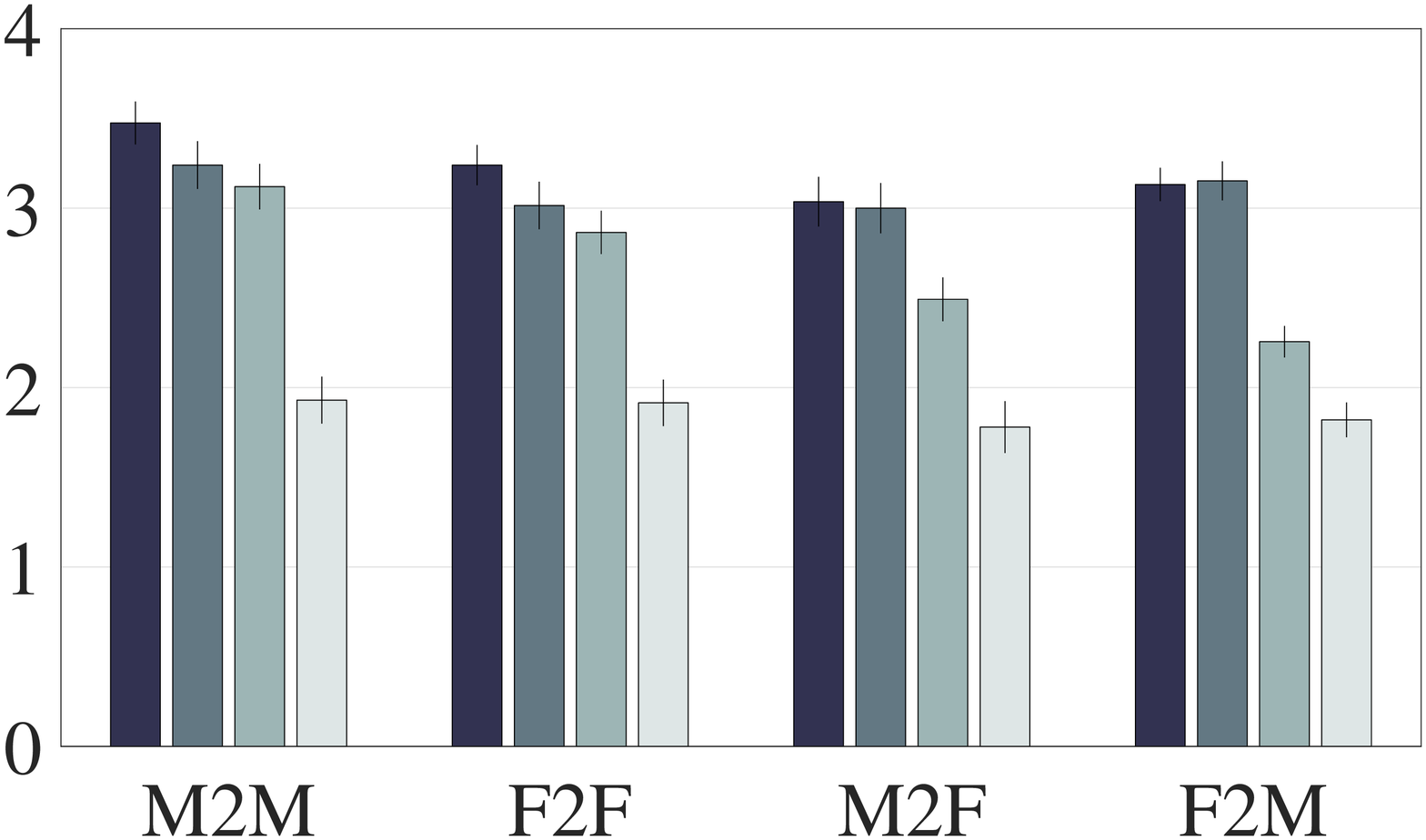}}
\subfigure[Similarity]{\includegraphics[width=0.481\linewidth]{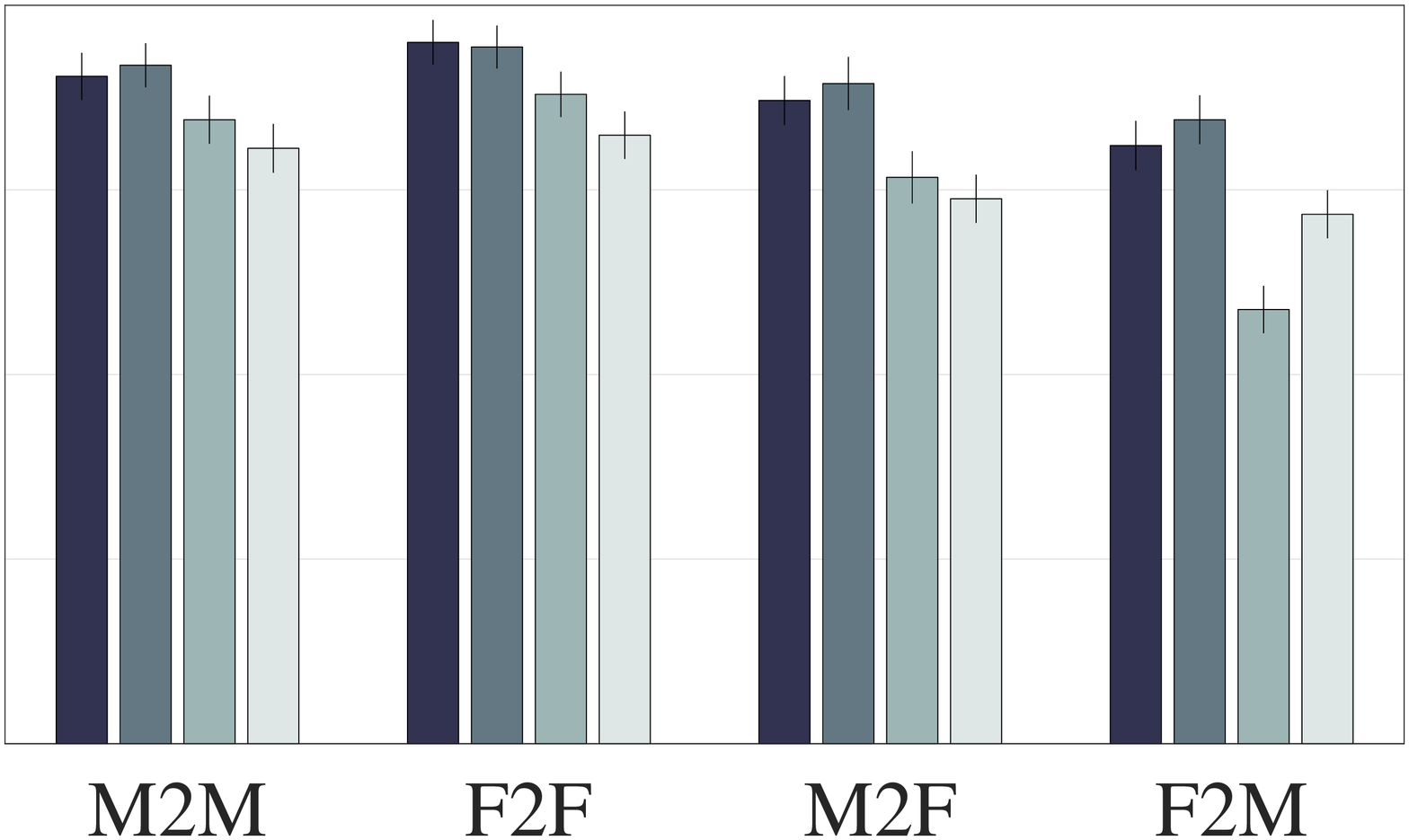}}
\caption{Subjective evaluation results for traditional conversion. The error bars denote 95\% confidence interval.}
\label{fig:exper1}
\end{figure}

Traditional many-to-many conversion task performs conversion only on speakers seen in the training corpus. Two baselines are compared with \algnamens, which we name StarGAN-VC \cite{kameoka2018stargan} and \citet{chou2018multi}. Both baselines are state-of-the-art in non-parallel many-to-many voice conversion. For \citet{chou2018multi}, we use the original implementation\footnote{https://github.com/jjery2243542/voice\_conversion} and its pre-trained model, which is trained on 20 speakers in the VCTK corpus. For fair comparison, the other models are trained on the same 20 speakers. Note that the training/test sets are partitioned differently from the Chou et al. pre-trained model, so we are giving the Chou et al. baseline an unfair advantage of seeing part of the test utterance during training. We use the open-source implementation for StarGAN-VC\footnote{https://github.com/liusongxiang/StarGAN-Voice-Conversion}.

\algname uses the speaker embeddings produced by the speaker encoder, while the baselines only use the one-hot embeddings of the speakers. To avoid unfair comparison and study if the performance advantage of \algname simply comes from the speaker embeddings, we implement another version of \algnamens, called \algnameonehotns, that also uses one-hot embeddings of the speakers.

To construct the utterances for the MTurk evaluation, $10$ speakers, $5$ male and $5$ female, are randomly chosen from the $20$ speakers in the training set. We then produce $10 \times 9 = 90$ conversions by converting a test utterance of each of the $10$ speakers to each of the $10$ speakers' voice. Each test unit, called HIT, contains conversion results of the same source-target speaker pair of the three algorithms, so there are $100$ HITs in total. Each HIT is assigned to $10$ subjects.

Fig. \ref{fig:exper1}(a) presents the MOS scores, and Fig.~\ref{fig:exper1}(b) presents the similarity scores. We are dividing the audio into four gender groups, male to male, male to female, female to male and female to female, and summarize the scores within each gender group. As shown in Fig. {fig:exper1}(a), the perceptual quality of the speech generated by \algname is much better than the baselines'. The MOS scores of \algname are above $3$ for all groups, whereas those for the baselines almost all fall below $3$. To give readers a better idea of what this means, notice that the MOS for 16kHz natural speech is around $4.5$. The MOS scores of the current state-of-the-art speech synthesizers are between $4$ and $4.5$ \cite{shen2018natural, arik2017deep}. The highest score in the 2016 Voice Conversion Challenge \cite{wester2016analysis} for \emph{parallel} conversion is $3.8$ for same-gender conversions, and $3.2$ for cross-gender conversion. Therefore, our subjective evaluation results show that \algname approaches the performance of parallel conversion systems in terms of naturalness, and is much better than existing non-parallel conversion systems.

In terms of similarity, \algname also out-performs the baselines. Note that for the baseline algorithms, there is a significant degradation from same-gender conversion to cross-gender conversion, but \algname algorithms do not display such a degradation. Finally, there is no significant difference between \algname and \algnameonehotns, which implies that the performance gain of \algname does not result from the use of the speaker encoder.

\subsection{Zero-Shot Conversion}
\label{subsec:zero-shot}

\begin{figure}[!t]
\centering
\includegraphics[width=1\columnwidth]{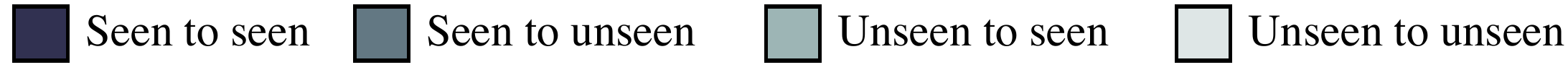}
\subfigure[MOS]{\includegraphics[width=0.5\linewidth]{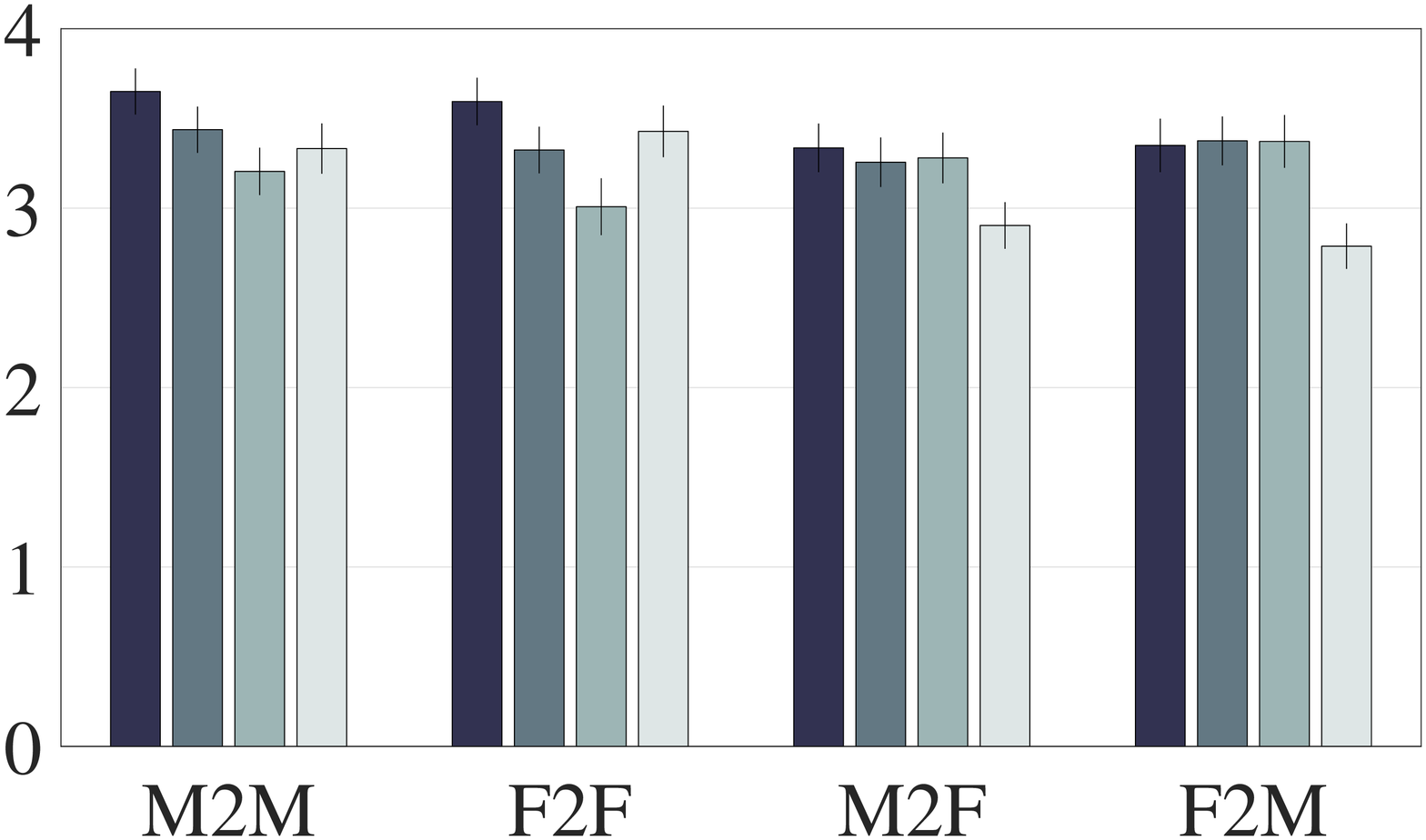}}
\subfigure[Similarity]{\includegraphics[width=0.481\linewidth]{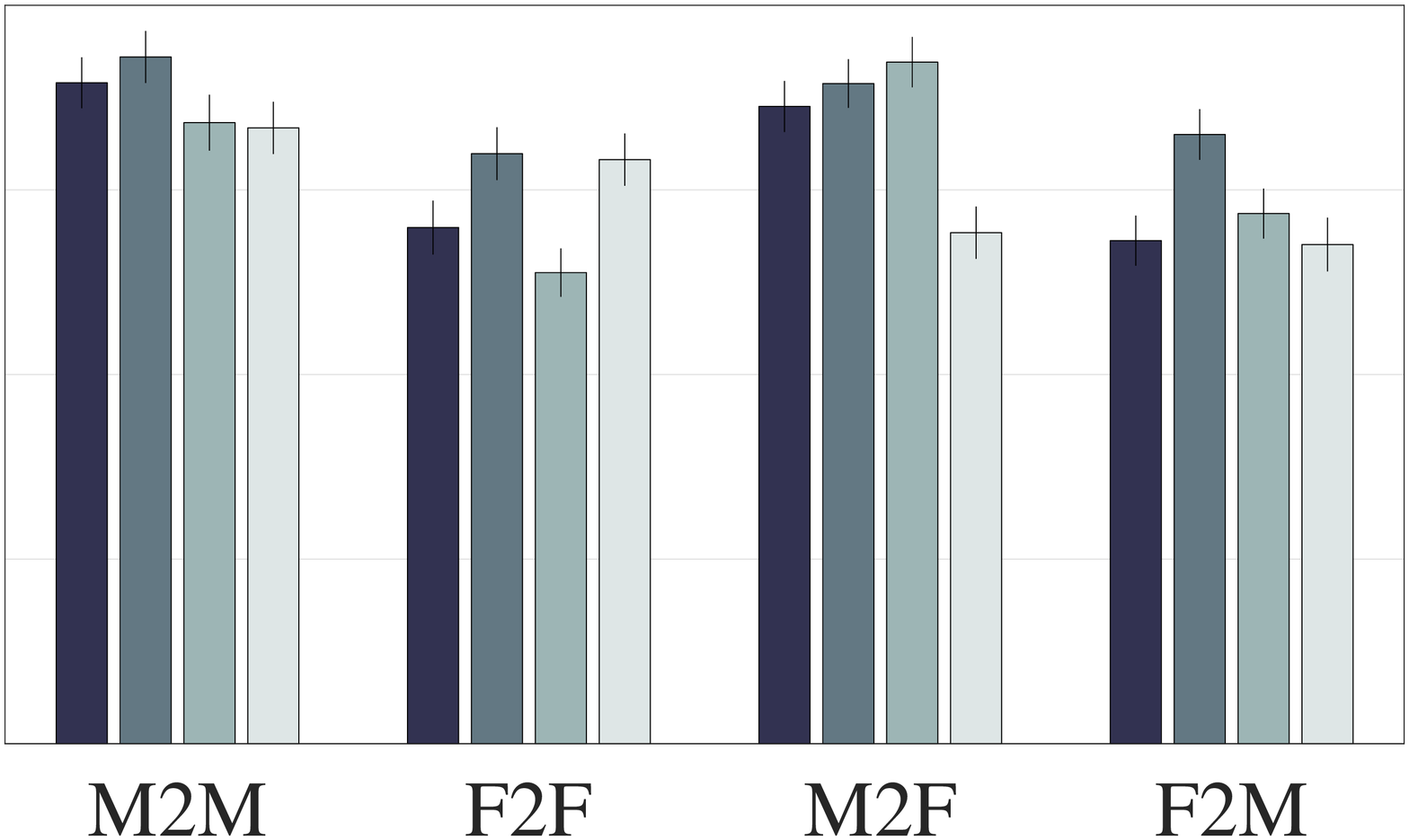}}
\caption{Subjective evaluation results for zero-shot conversion. The error bars denote 95\% confidence interval.}
\label{fig:exper2}
\end{figure}

Now we are ready to go beyond the traditional conversion task towards zero-shot conversion, where the target speakers are absent in the training set and only a few ($20$ seconds) utterances of each target speaker are available for reference. Since there are no zero-shot conversion baselines, we will compare the results within \algnamens.

The experiment settings are almost the same as in section~\ref{subsec:tradition}, except that the training set is expanded to $40$ speakers to improve the generalizability to unseen speakers. $10$ seen speakers and $10$ unseen speakers are selected for MTurk evaluation, so there are a total of $400$ source-target speaker pairs, each producing one conversion utterance. Each HIT contains four utterances, summing up to $100$ HITs in total. Each HIT is assigned to $10$ subjects.

Fig.~\ref{fig:exper2} presents the scores. There are three observations. First, for conversions among seen speakers, the performance is comparable to that in section~\ref{subsec:tradition}. Note that in this experiment, \algname is trained on $40$ speakers, which doubles the number of speakers used in the experiment in section~\ref{subsec:tradition}. Therefore, this comparable performance on seen speakers indicates that \algname is scalable to a large number of speakers in the training set.

Second, in terms of MOS score, \algname shows good generalizations to unseen speakers, with the MOS score exceeding $3$ in most settings. This means, even for unseen speakers, \algname is still able to outperform most existing non-parallel conversion algorithms.

Finally, in terms of the similarity score, there is an interesting observation that as long as seen speakers are included in either side of the conversions, the performance is comparable. There is a significant gap between conversions from unseen speakers to unseen speakers and the rest of the paradigms. Nevertheless, even for conversion between unseen speakers, which is the most challenging case, the similarity scores are still very competitive, which demonstrates \algnamens's competence in zero-shot conversion.

\subsection{Bottleneck Dimension Analysis}

\begin{table}[!t]
\caption{Assessment of the reconstruction quality and speaker disentanglement of \algnamens.}
\label{tab:3models}
\centering
	\begin{tabular}{ l| c c c }
    	\hline\hline
    	 & Narrow & \algnamens & Wide \\
    	\hline
    	Recon. Error & 34.6 & 8.59 & 3.85 \\
    	Class. Acc. & 7.50\% & 12.0\% & 70.5\% \\
    	\hline\hline
	\end{tabular}
\end{table}

Our theoretical justifications for the proposed style transfer autoencoder lies in the claim that the bottleneck dimension affects perfect reconstruction and disentanglement of content code and source speaker information, and that there exists a desirable bottleneck dimension where both properties hold (Fig.~\ref{fig:explain}). In this section, we will empirically validate this claim.

We measure \algnamens's reconstruction quality and the degree of disentanglement between the content code and the source speaker information. The reconstruction quality is measured by the $\ell_2$-norm of reconstruction error in the training set. Lower reconstruction error means higher reconstruction quality. The disentanglement is measured by training a speaker classifier on the content code and computing the classification accuracy on the training set. Higher classification accuracy means poorer disentanglement. The speaker classifier consists of $3$ fully-connected layers with 2,048, 1,024 and 1,024 hidden nodes respectively in each layer and softplus activation. The output activation is softmax and the training loss is cross entropy. The model architecture and experiment setting follow those in section~\ref{subsec:zero-shot}, so the speaker classification is on the $40$ seen speakers.

As references, we introduce two anchor models. The first model, which we name the ``too narrow'' model, reduces the dimensions of $C_{1\rightarrow}$ and $C_{1\leftarrow}$ from $32$ to $16$, and increases the downsampling factor from $32$ to $128$ (note that higher downsampling factor means lower temporal dimension). The second model, which we name the ``too wide'' model, increases the dimensions of $C_{1\rightarrow}$ and $C_{1\leftarrow}$ to $256$, and decreases the sampling factor to $8$, and $\lambda$ is set to $0$. Supposedly, according to Fig.~\ref{fig:explain}, the ``too narrow'' model should have low classification accuracy (good disentanglement) but high reconstruction error (poor reconstruction). The ``too wide'' model should have low reconstruction error (good reconstruction) but high classification accuracy (poor disentanglement). The normal \algname model should have both low reconstruction error (good reconstruction) and low classification accuracy (good disentanglement).

Table.~\ref{tab:3models} shows the reconstruction error and speaker classification accuracy for the three models. As expected, as the bottleneck dimension decreases, the reconstruction error increases and the classification accuracy decreases. What is interesting is that the normal \algname model does strike a good balance, with reconstruction error almost as low as the ``too wide'' model and the classification accuracy almost as low as the ``too narrow'' model. It is worth mentioning that \citet{chou2018multi} explicitly perform adversarial training to enforce speaker disentanglement. A similar classification experiment to test disentanglement is performed, and the classification accuracy is $45.1\%$ on $20$ speakers after the adversarial training is applied. In order to fairly compare with this result, we also perform a speaker classification test on the same $20$ speakers, and the classification accuracy is $14.2\%$. This result shows that bottleneck dimension tuning on speaker disentanglement is more effective than the more sophisticated adversarial training.
\section{Conclusion}
\label{sec:conclu}

In this paper, we have proposed \algnamens, a non-parallel voice conversion algorithm that significantly outperforms the existing state-of-the-art, and that is the first to perform  zero-shot conversions. In sharp contrast to its performance advantage is its simple autoencoder structure that trains only on self-reconstruction, and a bottleneck tuning to balance between reconstruction quality and speaker disentanglement. In an era of building increasingly sophisticated algorithms for style transfer, our theoretical justification and the success of \algname suggest that it is time to return to simplicity, because sometimes an autoencoder with a careful bottleneck design is all you need to make a difference.

\bibliography{example_paper}
\bibliographystyle{icml2019}

\clearpage

\section*{Appendix: Proving Thm~\ref{thm:guarantee}}

In this appendix, we will prove Thm~\ref{thm:guarantee}. We will start with a few lemmas:
\begin{lemma}
Given assumptions 1 and 2 in Thm.~\ref{thm:guarantee}, $S_1$ is a deterministic one-to-one mapping of $U_1$; so is $S_2$ to $U_2$.
\label{lem:one-to-one}
\end{lemma}
The proof is obvious and omitted.

\begin{lemma}
Given all the assumptions in Thm.~\ref{thm:guarantee}. Then exist a asymptotic global minimizer of Eq.~\eqref{eq:loss} that statisfies:
1. 
\begin{equation}
 \lim_{T \rightarrow \infty} \frac{1}{T} I(C_1; U_1) = 0
 \label{eq:lemma2}
\end{equation}
where $I(\cdot ; \cdot)$ denotes mutual information.

2.
\begin{equation}
    \plim_{T\rightarrow \infty} \hat{X}_{1\rightarrow1} = X_1
    \label{eq:lemma1}
\end{equation}

3.
\begin{equation}
    \lim_{T \rightarrow \infty} \mathbb{E} [(\hat{X}_{1\rightarrow1} - X_1)^2] + \lambda \mathbb{E} [(E_c(\hat{X}_{1\rightarrow1}) - C_1)^2]= 0
\end{equation}

\label{lem:independence}
\end{lemma}

\begin{proof}
Define the following set
\begin{equation}
    \mathcal{X} = \{x_1 : \log p_{X_1}(x_1 | U_1) \leq n-1 = n^* -1 + T^{2/3}\}
\end{equation}
$\mathcal{X}$ characterizes the set of instances where the optimal code length is guaranteed to be smaller than $n$.

Denote $C_1 = E^*_c(X_1; T)$ as the the following coding scheme. When $X_1 \in \mathcal{X}$, $C_1$ is the optimal lossless code for $p_{X_1}(\cdot | u_1)$ (whose code length is smaller than $n$ by Shannon's Coding Theorem) padded with $0$ to length $n$. When $X_1 \notin \mathcal{X}$, $C_1$ is any random number of dimension $n$.

Denote an auxiliary random variable
\begin{equation}
    A_1 = \mathbbm{1}[X_1 \in \mathcal{X}]
\end{equation}
where $\mathbbm{1}[\cdot]$ denotes the indicator function.

When $A_1=1$, there is a one-to-one mapping from $C_1$ to $X_1$, so we have
\begin{equation}
    H(C_1 | U_1, A_1=1) = H(X_1 | U_1, A_1=1)
    \label{eq:lossless_code}
\end{equation}
On the other hand, define $h_m$ as the capacity of each dimension of $C_1$, \emph{i.e.}
\begin{equation}
    h_\textrm{m} = \max_{p_{C_{1i}}(\cdot)} H(C_{1i})
\end{equation}
Then, the information $C_1$ contains is limited by the number of dimensions it has, \emph{i.e.}
\begin{equation}
\begin{aligned}
    H(C_1) &\leq \sum_i H(C_{1i}) \leq n h_\textrm{m} \\
    &\leq n^* h_m + T^{2/3} h_m \\
    &\leq H(X_1 | U_1) + 1 + T^{2/3} h_m
\end{aligned}
\label{eq:bottleneck_bound}
\end{equation}
where the second line is from assumption 3 of Thm.~\ref{thm:guarantee}. The third line is from the Shannon's coding theorem.

Notice that $A_1$ is a function of $X_1$, and thus we have
\begin{equation}
    \begin{aligned}
    H(X_1 | U_1) &= H(X_1, A_1 | U_1) \\
    &= H(X_1 | U_1, A_1) + H(A_1 | U_1) \\
    &\leq H(X_1 | U_1, A_1) + H(A_1) \\
    &= H(X_1 | U_1, A_1=1)p_{A_1}(1) \\
    &\quad + H(X_1 | U_1, A_1=0)p_{A_1}(0) + H(A_1) \\
    &= H(C_1 | U_1, A_1=1) p_{A_1}(1)\\
    &\quad + H(X_1 | U_1, A_1=0)p_{A_1}(0) + H(A_1) \\
    &\leq H(C_1 | U_1, A_1=1) p_{A_1}(1)\\
    &\quad + H(C_1 | U_1, A_1=0)p_{A_1}(0)\\
    &\quad + H(X_1 | U_1, A_1=0)p_{A_1}(0) + H(A_1) \\
    &= H(C_1 | U_1, A_1) \\
    &\quad + H(X_1 | U_1, A_1=0)p_{A_1}(0) + H(A_1) \\
    &\leq H(C_1 | U_1) \\
    &\quad + H(X_1 | U_1, A_1=0)p_{A_1}(0) + H(A_1)
    \end{aligned}
    \label{eq:x_bound}
\end{equation}
where the last but three line is given by Eq.~\eqref{eq:lossless_code}.

Eqs.~\eqref{eq:bottleneck_bound} and \eqref{eq:x_bound} imply that
\begin{equation}
\begin{aligned}
    I(C_1; U_1) &= H(C_1) - H(C_1 | U_1) \\
    & \leq 1 + T^{2/3} h_m \\
    &\quad + H(X_1 | U_1, A_1=0)p_{A_1}(0) + H(A_1)
    \label{eq:Hc_equality}
\end{aligned}
\end{equation}

For any $t \leq T$, $X_1(t)$ is a discrete random variable with finite support cardinality, denoted as $K$. Then we have
\begin{equation}
\begin{aligned}
    H(X_1 | U_1, A_1 = 0) &\leq \sum_{t=1}^T H(X_1(t) | U_1, A_1 = 0) \\
    & \leq T\log K
\end{aligned}
\label{x_ua_bound}
\end{equation}

On the other hand, notice that $\{X_1(t)\}$ is a stationary Markov process of order $\tau$. We have
\begin{equation}
\begin{aligned}
    \log p_{X_1}(\cdot | U_1) &= \sum_{t=1}^\tau \log p_{X_1(t)}(\cdot | U_1, X_1(1 : t-1)) \\
    & + \sum_{t=\tau+1}^T \log p_{X_1(t)}(\cdot | U_1, X_1(t-\tau : t-1)) \\
\end{aligned}
\end{equation}

From the central limit theorem for ergodic Markov process
\begin{equation}
    \lim_{T \rightarrow \infty}p_{A_1}(1) = 1, \quad \lim_{T \rightarrow \infty}p_{A_1}(0) = 0
    \label{eq:lln}
\end{equation}

Combining Eqs.~\eqref{eq:Hc_equality}, \eqref{x_ua_bound} and \eqref{eq:lln}, we have
\begin{equation}
\begin{aligned}
    \frac{1}{T} I(C_1; U_1) &\leq \frac{1}{T} (1 + T^{2/3} h_m + p_{A_1}(0)T\log K  + H(A_1)) \\
    &\rightarrow 0, \mbox{ as } T \rightarrow \infty
\end{aligned}
\end{equation}
Hence Eq.~\eqref{eq:lemma2} is proved.

Next, for $X_1 \in \mathcal{X}$, notice that $[C_1, S_1]$ is a lossless code of $[X_1, U_1]$, because
\begin{equation}
\begin{aligned}
    H(X_1, U_1) &= H(U_1) + H(X_1 | U_1) \\
    &= H(U_1) + H(C_1 | U_1) \\
    &= H(S_1) + H(C_1 | S_1) \\
    &= H(C_1, S_1)
\end{aligned}
\label{eq:lossless_code2}
\end{equation}
where the second line is from Eq~\eqref{eq:lossless_code}; the third line is from Lem.~\ref{lem:one-to-one}. Eq.~\eqref{eq:lossless_code2} implies that $[U_1, X_1]$ is fully recoverable from $[C_1, S_1]$. Therefore, there exists an optimum decoder $D^*(\cdot, \cdot)$ such that
\begin{equation}
    \hat{X}_{1 \rightarrow 1} = X_1
    \label{eq:equal_x}
\end{equation}
Combining Eqs.~\eqref{eq:lln} and \eqref{eq:equal_x}, Eq.~\eqref{eq:lemma1} is proved.

Apply $E_c(\cdot)$ to both sides, we get
\begin{equation}
    \plim_{T \rightarrow \infty} E_c(\hat{X}_{1 \rightarrow 1}) = E_c(X_1) = C_1
\end{equation}
Hence, considering $X_1$ has finite second moment, convergence with probability implies mean squared convergence, \emph{i.e.}
\begin{equation}
    \lim_{T \rightarrow \infty} \mathbb{E} [\lVert \hat{X}_{1\rightarrow1} - X_1)^2 \rVert_2^2] + \lambda \mathbb{E} [\lVert E_c(\hat{X}_{1\rightarrow1}) - C_1 \rVert_1]= 0
    \label{eq:perfect_loss}
\end{equation}
which means that $[E^*_c(\cdot), D^*(\cdot, \cdot)]$ is the asymptotic global optimizer of Eq.~\eqref{eq:loss}.
\end{proof}

Now we are ready to prove Thm~\ref{thm:guarantee}.

\begin{proof} (Thm.~\ref{thm:guarantee})
Denote $X'_2$ as speech drawn from the ground truth distribution of the converted speech, \emph{i.e.} $p_X(\cdot | U = U_2, Z = Z_1)$. Then our goal is to show that $\hat{X}_{1 \rightarrow 2}$ is assymptotically identically distributed to $X'_2$.

What we will do is bridge the two random variables by passing $X'_2$ to \algname for self-reconstruction. Namely,
\begin{equation}
    C'_2 = E^*_c(X'_2), \mbox{ and } \hat{X}'_{2\rightarrow2} = D^*(C'_2, S_2)
\end{equation}
where $E^*(\cdot)$ and $D(\cdot)$ are the optimal encoder and decoder derived in Lem.~\ref{lem:independence}.

From Lem.~\ref{lem:independence}, we know that $\hat{X}'_{2\rightarrow2} \rightarrow X'_2$ with probability. So all is left to do is to show that $\hat{X}'_{2\rightarrow2}$ is assymptotically identically distributed to $\hat{X}_{1\rightarrow2}$.

First, notice that
\begin{equation}
\begin{aligned}
    p_{C_1}(\cdot | z_1, u_2) &= p_{C_1}(\cdot | z_1) \\
    &= p_{E_c(X_1)}(\cdot | z_1) \\
    &= p_{E_c(X)}(\cdot | Z = z_1)
\end{aligned}
\label{eq:c_id1}
\end{equation}
where the first line is due to the fact that $C_1$ and $Z_1$ are both independent of $U_2$ (Recall $U_2$ is not involved in the generation process of $C_1$); the last line is from the fact that $(U_1, Z_1, X_1)$ is identically distributed to $(U, Z, X)$.

Therefore, we can show that
\begin{equation}
\begin{aligned}
    &\lim_{T\rightarrow \infty} \frac{1}{T} KL(p_{C'_2}(\cdot | z_1, u_2) || p_{E_c(X)}(\cdot | Z = z_1)) \\
    =&\lim_{T\rightarrow \infty} \frac{1}{T} KL(p_{C'_2}(\cdot | z_1, u_2) || p_{C_1}(\cdot | Z = z_1)) \\
    =& 0
\label{eq:c_id2}
\end{aligned}
\end{equation}
where the last line is given by Eq.~\eqref{eq:lemma2} of Lem.~\ref{lem:independence}.


On the other hand,
\begin{equation}
\begin{aligned}
    p_{\hat{X}_{1\rightarrow2}}(\cdot | z_1, u_2) = p_{D^*(C_1, S_2)} ( \cdot | z_1, u_2) \\
    p_{\hat{X}'_{2\rightarrow2}}(\cdot | z_1, u_2) = p_{D^*(C'_2, S_2)} ( \cdot | z_1, u_2)
\end{aligned}
\label{eq:ideal_decomp}
\end{equation}

Combining Eqs.~\eqref{eq:c_id2} and \eqref{eq:ideal_decomp}, we have
\begin{equation}
    \lim_{T \rightarrow \infty} \frac{1}{T} KL(p_{\hat{X}_{1\rightarrow2}}(\cdot | z_1, u_2) || p_{\hat{X}'_{2\rightarrow2}}(\cdot | z_1, u_2)) = 0
\end{equation}
\end{proof}

Here is a final note on Thm~\ref{thm:guarantee}. The content loss can help to constrain information capacity of the bottleneck by soft-constraining the range of each dimension of the content code, otherwise the information capacity of each bottleneck dimension can be unbounded and Thm~\ref{thm:guarantee} does not apply.

\end{document}